\documentclass[conference]{IEEEtran}
\IEEEoverridecommandlockouts
\usepackage[utf8]{inputenc} 
\usepackage[T1]{fontenc}
\usepackage{ifthen}
\usepackage{cite}
\usepackage{comment}
\usepackage{amsmath,amssymb,amsfonts}
\usepackage{algorithmic}
\usepackage{graphicx}
\usepackage{textcomp}
\usepackage{xcolor}
\usepackage{enumitem}

\newtheorem{theorem}{Theorem}
\newtheorem{definition}{Definition}
\newtheorem{lemma}{Lemma}
\usepackage{graphicx}
\usepackage{nicematrix}

\newcommand{\Adn}{\mathbf{A}_{n,d}}

\begin{document}

\title{Order Optimal Task Allocation in Distributed Computing via Interweaved Cliques}

\author{\IEEEauthorblockN{Javad Maheri, K. K. Krishnan Namboodiri, and Petros Elia}
\IEEEauthorblockA{EURECOM, Sophia Antipolis, France \\
Email: \{maheri, karakkad, elia\}@eurecom.fr}\thanks{This work was supported by the Huawei France-funded Chair towards Future Wireless Networks, by the French government under the France 2030 ANR program “PEPR Networks of the Future” (ref. ANR-22-PEFT-0010), and by European Research Council ERC-StG Project SENSIBILITÉ under Grant 101077361.}
}

\maketitle
\begin{abstract}
We consider a distributed computing system in which a master node coordinates $N$ workers to evaluate a function over $n$ input files, where this function accepts general decomposition. In particular, we focus on the general case where the requested function admits a $d$-uniform decomposition, meaning that it can be decomposed into a set of subfunctions that each depends on a unique $d$-tuple of the $n$ files.
Our objective is to design file and task allocations that minimize the worst-case communication from the master to any worker and the worst-case computational load across workers. 
We first show that the optimal file and task allocation with minimum communication and computation costs admits a natural characterization within combinatorial design theory: it corresponds to a Steiner system $S(t, k, v)$ with $t=d$, $v=n$, and $k \approx \frac{n}{N^{1/d}}$. However, Steiner systems are known to exist only for very restricted parameter regimes.
To overcome this limitation, we propose the information-theoretic-inspired \emph{Interweaved Clique (IC) design}, a universal and deterministic allocation framework that relaxes the strict structure of Steiner systems by allowing slight variations in worker file loads. Although slightly suboptimal, the IC design achieves a communication cost within a constant factor $4e$ from our converse, while also maintaining an order-optimal computation cost, thus allowing this work to derive the fundamental scaling laws of this general distributed computing problem for a large range of parameters.
\end{abstract}


\begin{IEEEkeywords}
Distributed Computing, Coded Distributed Computing, Coded Caching, Distributed Learning, Communication Optimization, Map Reduce, Hypergraph Partitioning, Combinatorial Designs, $t$-Designs, Steiner Systems.
\end{IEEEkeywords}

\section{Introduction}
\label{sec:intro}

The efficient allocation of computational tasks and data is a cornerstone of modern distributed computing, caching, and distributed learning frameworks~\cite{MaN,8963629,bitar2020stochastic, 8007060,vavilapalli2013apache}. Across applications such as large-scale machine learning, covariance matrix estimation, and scientific simulation, system performance is often constrained by the volume of communication required during distributed execution and by the associated computational burden. This challenge has motivated extensive recent work on communication- and computation-efficient distributed function evaluation. A prominent line of research studies coded distributed computing frameworks, beginning with Coded MapReduce~\cite{LMYA} and extending to variants that address stragglers, heterogeneity, and network topology~\cite{CCW,WaW,BWW,PNR,BrEl,WSJTC,YaJ,MaT}, demonstrating that structured data placement and task allocation can substantially reduce communication through coded exchanges. Another set of works focuses on linearly separable functions and straggler resilience~\cite{WSJC,WSJC2}, multi-user architectures and task assignment using covering and tiling constructions~\cite{KhE,KhE2}, and worst-case communication minimization under task constraints~\cite{NPME}. All the above lines of research share the goal of designing, under various settings and assumptions, task and data assignment methods that reduce communication and computation costs in distributed computing, and in certain settings, highlight inherent interactions between these two resources.

Motivated by the same need for efficient task and data allocation, we here consider a general coded distributed computing framework for computing decomposable functions in distributed systems. In our framework, the function is decomposable into multiple subfunctions, each taking as input a different $d$-tuple of files, allowing the master node to assign collections of subfunctions (each represented here by a $d$-tuple) to multiple workers and communicates the necessary file inputs so as to enable local computation. Naturally, this setting entails a communication cost (as servers need to be communicated the necessary files), and a computation cost (which scales with the number of subfunctions each server must compute). The central design problem is to jointly determine the task assignment and file placement strategies that minimize communication while maintaining balanced computational loads across workers. Unlike formulations tailored to specific computation pipelines, our framework models the distributed evaluation of general \(d\)-tuple decomposable functions by explicitly characterizing the interaction between task partitioning and file placement. The resulting problem is inherently combinatorial in nature. Accordingly, we seek solutions based on structured combinatorial constructions that minimize both communication and computation costs.

Indeed, combinatorial designs have been widely used in prior coded distributed computing frameworks, particularly in MapReduce-based models, where clique covers based on \(t\)-designs and related combinatorial structures guide task and data assignment~\cite{CDC_Design1,CDC_Design2,CDC_Design3,CDC_Design4, maheri2025constructing}. While these approaches have been effective in reducing communication, they are largely specialized to MapReduce-style computation pipelines. In contrast, our work develops combinatorial constructions for a more general \(d\)-tuple subfunction model, enabling distributed computation beyond the MapReduce paradigm.

\emph{Notations:} We represent \(d\)-tuples using bold lowercase letters, such as \(\mathbf{a} = \{a_1, a_2, \dots, a_d\}\). Sets of \(d\)-tuples are denoted by bold uppercase letters, such as \(\mathbf{\Phi}\). We use \([n]\) to denote the set \(\{1,2,\ldots,n\}\). For any set \(\mathcal{S}\), \(\binom{\mathcal{S}}{d}\) denotes the set of all \(d\)-sized subsets of \(\mathcal{S}\). We use \(\mathbf{A}_{n,d}\) to denote \(\binom{[n]}{d}\). For positive integers \(a\) and \(b\), \(a \mid b\) indicates that \(a\) divides \(b\). Finally, we write \(f(n)\asymp g(n)\) if there exist constants \(c_1, c_2\), and \(n_0\) such that for all \(n \ge n_0\), \(c_1 g(n) \le f(n) \le c_2 g(n)\).

\section{System Model and Problem Statement}
\label{sec:system_model}
We consider a distributed computing system consisting of \(N\) worker nodes (servers) and a master node that coordinates the computation of a desired function of \(n\) input files (as shown in Fig. \ref{SystemModel}). The master has access to a library of files
\(
\mathcal{W}=\{W_1,\ldots,W_n\},
\)
where each file \(W_j\in\mathbb{F}^B\) contains \(B\) symbols over a field
\(\mathbb{F}\). The desired function
\(
F:(\mathbb{F}^{B})^n \rightarrow \mathbb{F}^L
\)
is assumed to be decomposable into \(\binom{n}{d}\) subfunctions, each depending on a unique subset of \(d\) files. Accordingly, any such decomposition can be written as
\begin{equation} \label{generalFunction1}
\Psi\big(\{\zeta_\mathcal{T}(\mathcal{W}_\mathcal{T}):\mathcal{T}\in
\binom{[n]}{d}\}\big):
(\mathbb{F}^T)^{\binom{n}{d}} \rightarrow \mathbb{F}^L
\end{equation}
where \(\Psi\) is an aggregation function and each subfunction
\(\zeta_\mathcal{T}:(\mathbb{F}^B)^d \rightarrow \mathbb{F}^T\) operates on the set of files \(\mathcal{W}_\mathcal{T}=\{W_j:j\in\mathcal{T}\}\). The parameter \(d\) is
referred to as the \emph{subfunction file degree}.

\emph{Task and File Allocation:} The master assigns to each worker \(b\in[N]\) a set of subfunctions \(\mathbf{\Phi}_b \subseteq \mathbf{A}_{n,d}\) to compute. The collection
\(
\mathbb{P}\triangleq\{\mathbf{\Phi}_1,\ldots,\mathbf{\Phi}_N\}
\)
forms a partition of \(\mathbf{A}_{n,d}\),
\begin{equation}\label{eq: partition}
\bigcup\nolimits_{b=1}^N \mathbf{\Phi}_b = \mathbf{A}_{n,d}, \qquad
\mathbf{\Phi}_b \cap \mathbf{\Phi}_{b'}=\emptyset \;\; \text{for } b\neq b'.
\end{equation}
To compute the subfunctions in \(\mathbf{\Phi}_b\), worker \(b\) must receive all files indexed by the union of the \(d\)-tuples in \(\mathbf{\Phi}_b\). Let
\begin{equation}\label{eq: alpha}
\alpha(\mathbf{\Phi}_b)=\bigcup\nolimits_{\mathcal{T}\in\mathbf{\Phi}_b}\mathcal{T}
\subseteq[n]
\end{equation}
denote the set of required file indices. The master sends the set of files \(\mathcal{W}^{(b)}=\alpha(\mathbf{\Phi}_b)\) to worker \(b\). The communication cost is defined as
\begin{equation}\label{eq: pi}
\pi=\max_{b\in[N]}|\alpha(\mathbf{\Phi}_b)|
\end{equation}
which captures the bottleneck communication load across workers.

\emph{Computing Phase:} During the computing phase, each worker \(b\) computes all subfunctions
\(\zeta_\mathcal{T}(\mathcal{W}_\mathcal{T})\) for
\(\mathcal{T}\in\mathbf{\Phi}_b\). Assuming identical computational capability across workers and equal cost per subfunction, the computation time is proportional to the maximum number of subfunctions assigned to any worker. Accordingly, the computation cost is defined as
\begin{equation}\label{eq: delta}
\delta=\frac{\max_{b\in[N]}|\mathbf{\Phi}_b|}
{\lceil \binom{n}{d}/N\rceil}
\end{equation}
where the denominator corresponds to the ideal uniform assignment achieving the minimum possible computation delay.

\begin{figure}
\begin{center}   
\includegraphics[width=8cm]{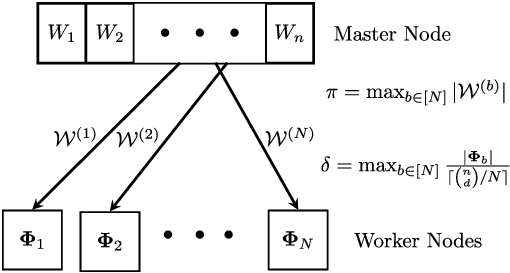}
\caption{Distributed computing model.
The set of files communicated to worker \(b\) is denoted with \(\mathcal{W}^{(b)}=\alpha(\mathbf{\Phi}_b)\), while \(\pi\) denotes the maximum communication cost across the master--worker links under the assumption of parallel, uniform-capacity links. Similarly, the set of subfunctions assigned for computation to worker \(b\) is denoted by \(\mathbf{\mathbf{\Phi}}_{b}\), while \(\delta\) denotes the computational delay normalized by the minimum possible computational delay, assuming homogeneous workers.}
\label{SystemModel}
\end{center}
\vspace{-0.5cm}
\end{figure}
We define \(\pi^\star\) as the minimum communication cost over all valid allocation schemes. Our objective is to design a partition \(\mathbb{P}=\{\mathbf{\Phi}_b\}_{b=1}^N\) that minimizes the communication cost \(\pi\) while ensuring that the computation cost \(\delta\) remains close to unity, for given \(n,d,\) and \(N\).

Our approach, and specifically the \(d\)-tuple based decomposability model we proposed, is motivated by the many practically relevant functions that require aggregating interactions over a large fraction of subfunctions - and which thus entail interactions over the \(d\)-tuples in \(\mathbf{A}_{n,d}=\binom{[n]}{d}\), the collection of all subsets of \([n]\) of size \(d\). Canonical examples include covariance and correlation computations involving all pairwise dependencies (\(d=2\))~\cite{ledoit2004well}, higher-order cumulant estimation with \(d>2\) interactions~\cite{comon1994independent}, and kernel matrix construction in kernel methods, which requires evaluating similarities over all pairs of data points~\cite{scholkopf1998nonlinear}. Similar dense interaction patterns arise in particle and molecular dynamics simulations, where forces are computed over all particle pairs or higher-order groups~\cite{Dhaliwal2022RandomFeaturesInteratomic,rahimi2007random}, as well as in exhaustive SNP--SNP interaction analysis in genomics~\cite{li2015overview}. 
In all these settings, the final output depends on a very large collection of \(d\)-tuple subfunctions, making scalable distributed computation a central challenge.

\subsection{The Combinatorial Perspective: Steiner Systems}

The set of all \(d\)-tuples (the subfunctions) can be represented as the edge set of a complete \(d\)-uniform hypergraph on vertex set \([n]\). Consequently, assigning subfunctions to \(N\) servers is equivalent to partitioning the hyperedges of this complete hypergraph into \(N\) groups.
From combinatorial design theory, one natural solution to this hyperedge partitioning problem is a \textit{Steiner system}, denoted by \(S(t,k,v)\).
\begin{definition}[Steiner System]
A Steiner system \(S(t,k,v)\) consists of a set of \(v\) points and a collection of blocks, each of size \(k\), such that every subset of \(t\) points is contained in exactly one block.
\end{definition}

Directly from the definition, the following lemma follows.

\begin{lemma}
\label{lem:steiner}
Let \(S(t,k,v)\) be a Steiner system. Consider a distributed computing setting with \(n=v\) files and subfunction degree \(d=t\), and let \(N=\binom{n}{d}/\binom{k}{d}\). Then there exists an assignment of files and subfunctions to \(N\) workers such that each worker is communicated with
\(k\) files and is assigned exactly \(\binom{k}{d}\) subfunctions. Consequently, the communication cost \(\pi=k\) and the computation cost \(\delta=1\) are achievable.
\end{lemma}
\begin{IEEEproof}
Let \(\mathcal{S}_1,\ldots,\mathcal{S}_N\) denote the blocks of \(S(d,k,n)\). Assign worker \(b\) the subfunction set \(\mathbf{\Phi}_b=\binom{\mathcal{S}_b}{d}\), and communicate to it the \(k\) files indexed by \(\mathcal{S}_b\). Since every
\(d\)-subset of \([n]\) is contained in exactly one block, the sets \(\{\mathbf{\Phi}_b\}_{b=1}^N\) form a partition of \(\mathbf{A}_{n,d}\), and each worker is assigned \(\binom{k}{d}\) subfunctions. Hence, we have \(\pi=k\) and \(\delta=1\).
\end{IEEEproof}

However, the applicability of Steiner systems is severely limited by their stringent existence requirements. A Steiner system \(S(d,\pi,n)\) can exist only if the divisibility conditions that \(\binom{n-i}{d-i}\) is divisible by
\(\binom{\pi-i}{d-i}\) for all \(0 \le i\le d-1\) are satisfied, and even when these necessary conditions hold, existence is guaranteed only for restricted parameter regimes. Consequently, for most choices of \(n\), \(d\), and \(N\), no Steiner system exists. This strong structural rigidity makes Steiner systems unsuitable as a general design tool for distributed computing systems in which \(n\) or \(N\) may vary freely.

\section{Main Results}

To overcome the non-existence of Steiner systems for most parameter regimes, we propose the \emph{Interweaved Clique (IC) design}, a relaxed combinatorial framework that provides a constructive and broadly applicable solution for a wide range of system parameters $(n, N, d)$. The relaxation allows a slight variation in the number of tasks assigned to each worker, while preserving a deterministic structure based on interweaved cliques. Unlike prior appearances of clique-based constructions in information-theoretic problems such as coded caching~\cite{MaN}, where cliques represent user side-information structures, in this work cliques serve as the combinatorial seed for constructing the worker--file--task allocation. Specifically, the global allocation is built over $\mathbf{A}_{n,d}$ using a smaller seed structure $\mathbf{A}_{f,d}$, where $f$ can be much smaller than $n$. Similar two-level complete-set formulations have also appeared in other coding-theoretic and information-theoretic contexts~\cite{MKR,BrEl2,BrEl,8437333,NaR,PNR2,EngEli2017,ZhaoBazcoElia2023}.\\
The proposed IC design enjoys the following key properties:
\subsubsection{\textbf{Universality}}{\setlength{\parindent}{0pt}A valid partition exists for a wide range of parameters $(n, N, d)$.
\subsubsection{\textbf{Order-Optimality}}The proposed scheme achieves $\pi \leq 4e\,n/N^{1/d}$ and therefore achieves a gain in communication cost that scales as $N^{1/d}$, where $e$ denotes Euler's constant. Moreover, the scheme is order-optimal, as its communication cost is within a constant factor $4e$ of the information-theoretic lower bound obtained from a packing argument.


Let us now formally state the performance guarantees of the IC design, emphasizing the communication cost $\pi$.

\begin{theorem}[Optimal Communication Cost]
\label{th: main}
For a distributed computing system with $n$ files and $N \le (\frac{9}{10}\sqrt{\frac{n}{d}})^d$ workers, the IC design achieves a communication cost $\pi$ satisfying
\[
  \pi\le \frac{4e \cdot n}{N^{1/d}}.
\]
Furthermore, the scaling law $\pi \asymp \frac{n}{N^{1/d}}$ is optimal.
\end{theorem}

\begin{IEEEproof}
The proof of Theorem~\ref{th: main} follows from the IC design presented in Section~\ref{sec:IC_design} and a converse on \(\pi^\star\) presented in Section~\ref{sec:lb}, and is provided in Section~\ref{Proof:Main}.
\end{IEEEproof}
In the end, let us here also note that in the extended version of our work in~\cite{IC}, we also show that this same IC design guarantees, in a very broad setting and with high probability, a near-optimal computation cost of $\delta \le 4.$ 

\section{Achievable Scheme: Interweaved-Cliques Design}
\label{sec:IC_design}


We now describe the Interweaved-Cliques (IC) design, which constructs a partition of $\Adn$ for any given tuple \((n,d,N\leq (\frac{9}{10}\sqrt{\frac{n}{d}})^d )\). The design leverages an intermediate parameter $f$ to structure the file library into \emph{families} and then assigns tasks based on the intersection of these families. The parameter $f$ is choosen as
\begin{equation}
\label{eq: largest r}
    f = \max \big\{ r \in \mathbb{Z}^+ \mid \binom{r}{d} \le N \big\}.
\end{equation}
First, we design the partition of \(\mathbf{A}_{n,d}\) for an intermediate number of groups \(N'\), defined as \(N' \triangleq \binom{f}{d}\). Then, we extend the construction from \(N'\) to \(N\) groups. The construction of \(\mathbf{\Phi}_1, \mathbf{\Phi}_2, \ldots, \mathbf{\Phi}_{N'}\) proceeds in two cases.
\subsection{Case 1: Divisible Parameters ($f\mid n$)}
{\setlength{\parindent}{0pt}
\label{subsec: case1}
Assume $n$ is divisible by $f$, so $s \triangleq n/f$ is an integer.

\subsubsection{File Families}
We partition the $n$ files into $f$ disjoint sets called families, denoted $\mathcal{F}_1, \dots, \mathcal{F}_f$, each containing $s$ files. Specifically, $\mathcal{F}_i$ contains files with indices $\{(i-1)s+1, \dots, is\}$.

\subsubsection{Group Identification}
We create $N' = \binom{f}{d}$ base groups, indexed by the set of $d$-subsets of families, i.e., $\sigma \in {[f]\choose d}$. We denote the set of all base groups as
\begin{equation}
\label{eq: big sigma}
   \mathbf{\Sigma} \triangleq \left\{ \sigma \subseteq [f] : |\sigma| = d \right\}.
\end{equation}
Each group \(\sigma \in \mathbf{\Sigma}\) is allocated the union of files from the families in \(\sigma\). The number of files allocated to each base group is thus
\begin{equation}
\pi = \sum_{i \in \sigma} |\mathcal{F}_i| = s \cdot d = \frac{n \cdot d}{f}.
\end{equation}

\subsubsection{Task Partitioning}
We partition $\Adn$ by assigning each $d$-tuple $\mathbf{a} \in \Adn$ to one of the groups. We define the \textit{support family} of a tuple $\mathbf{a}$ as $\mathcal{B}(\mathbf{a}) \triangleq \{ j \in [f] \mid \mathbf{a} \cap \mathcal{F}_j \neq \emptyset \}$.
\begin{itemize}[leftmargin=*]
    \item \textbf{Full Support Tuples:} We let \( \mathbf{A}_{\mathrm{full}} \triangleq \{ \mathbf{a}\in \mathbf{A}_{n,d} \mid  |\mathcal{B}(\mathbf{a})| = d\}\) represent the so-called set of full support (maximal support) $d$-tuples. If $|\mathcal{B}(\mathbf{a})| = d$, the \(\mathbf{a}\) intersects exactly one file from \(d\) distinct families. For a \(d\)-tuple $\mathbf{a} \in \mathbf{A}_{\mathrm{full}}$, let $  \mathcal{B}(\mathbf{a})=\sigma$. We assign $\mathbf{a}$ to the group $\sigma$ which belongs to $[N']={[f] \choose d}$.  For each $ \sigma \in {[f] \choose d}$, we form its full supports members as follows
\begin{equation}
\label{eqfull}
    \mathbf{\Phi}^{\mathrm{(full)}}_\sigma\triangleq
    \{\mathbf{a}\in \mathbf{A}_{n,d}\ |\ \mathcal{B}(\mathbf{a})=\sigma\}.
\end{equation}
    This forms the \emph{clique} core of the design. 
    \item \textbf{Complement Tuples:} We also let \(\mathbf{A}_{\mathrm{com}}\triangleq\mathbf{A}_{n,d}\backslash \mathbf{A}_{\mathrm{full}}\). For each $\mathbf{a} \in \mathbf{A}_{\mathrm{com}}$, $|\mathcal{B}(\mathbf{a})| < d$. This means, this $d$-tuple is supported by a subset of families $\mathcal{I} \subset [f]$ with $|\mathcal{I}| < d$. 
    The \(\mathbf{a}\) is eligible to be assigned to all \(\sigma \in {[f] \choose d}\) such that \(\mathcal{I} \subset \sigma\).  
   To design partition \(\mathbf{A}_{\mathrm{com}}\) into \(N'\) groups, we consider a partition of \(\mathbf{A}_{n,d}\), which classifies its \(d\)-tuples according to the size of their support family. This partition is as follows
\begin{equation}
    \mathbf{A}_{n,d}=\{\mathbf{C}_{\lceil{\frac{d}{s}}\rceil},\mathbf{C}_{\lceil{\frac{d}{s}}\rceil+1},\dots,\mathbf{C}_d\}
\end{equation}
where for each \(\beta\in [\lceil{\frac{d}{s}}\rceil,d]\), the set
$\mathbf{C}_\beta \triangleq \{\mathbf{a}\in \mathbf{A}_{n,d}\mid |\mathcal{B}(\mathbf{a})|=\beta\}$ represents the set of \(d\)-tuples \( \mathbf{a}\in\mathbf{C}_\beta \) that each intersects exactly \( \beta \) families. Naturally, we have \(\mathbf{C}_d= \mathbf{A}_{\mathrm{full}}\) and \(\mathbf{A}_{\mathrm{com}}=\bigcup_{\beta=\lceil\frac{d}{s}\rceil}^{d-1}\mathbf{C}_\beta\). Let us fix a \(\beta \in [\lceil{\frac{d}{s}}\rceil,d-1]\). Then, for each \( \mathcal{I} \in \binom{[f]}{\beta} \), we define
    \begin{equation}
    \label{eq:C_beta, I}
          \mathbf{C}_{\beta,\mathcal{I}} \triangleq \left\{ \mathbf{a} \in \mathbf{C}_\beta \;\middle|\; \mathcal{B}(\mathbf{a})=\mathcal{I}\right\}.
    \end{equation}
We define \(\mathbf{C}_{\beta,\mathcal{I},\sigma}\) as the subset of \(\mathbf{C}_{\beta,\mathcal{I}}\) allocated to group \(\sigma\). In Appendix~B.C of \cite{IC}, we describe a sequence of steps that leads to the construction of the sets \(\mathbf{C}_{\beta,\mathcal{I},\sigma}\). For each \(\sigma \in \binom{[f]}{d}\), there exist \(\binom{d}{\beta}\) distinct \(\mathcal{I} \in \binom{[f]}{\beta}\) such that \(\mathcal{I} \subset \sigma\). Consequently, 
$\mathbf{\Phi}_{\sigma}^{\mathrm{(com)}}\triangleq\bigcup_{\beta=\lceil \frac{d}{s}\rceil}^{d-1}\bigcup_{\mathcal{I}\subset \sigma}\mathbf{C}_{\beta, \mathcal{I}, \sigma}.$ 
Then, the subfunctions (\(d\)-tuples) allocated to worker \(\sigma\), where \(\sigma \in {[f]\choose d}\), is
 \begin{equation}
     \label{eq:SubfunctionCase1}
     \mathbf{\Phi}_{\sigma}\triangleq\mathbf{\Phi}_{\sigma}^{(\mathrm{full})}\cup \mathbf{\Phi}_{\sigma}^{\mathrm{(com)}}.
 \end{equation}
\end{itemize}


\subsection{Case 2: General Parameters ($f \nmid n$)}
\label{subsec: Case 2}
In this case, we cannot create equal-sized families. We adapt the design by introducing \emph{excluded} files.
Let $s_0 \triangleq \lfloor \frac{n}{f+d} \rfloor + 1$ and define $g \triangleq n - f \cdot s_0$. We treat the last $g$ files as an excluded set $\mathcal{E}$, i.e.,
\begin{equation}
\label{eq: excluded elements}
\mathcal{E}\triangleq [n]/[n-g]=\{n, n-1, \ldots, n-g+1\}. 
\end{equation}
The remaining $n' \triangleq n-g$ files are partitioned into $f$ families of size $s_0$. The partition of $\Adn$ is constructed by
\subsubsection{Step 1} In this step, we apply the Case \ref{subsec: case1} construction to the $n'$ non-excluded files, i.e., $[n']$. Thus, each group \(\sigma \in {[f] \choose d}\) receives from \(\mathbf{A}_{n',d}\)
\begin{equation}
    \label{eq: step 1 portion} 
\mathbf{\Phi}_{\sigma}^{(\mathrm{full})}\cup \mathbf{\Phi}_{\sigma}^{\mathrm{(com)}}.
\end{equation}
\subsubsection{Step 2}  The second step considers the excluded \footnote{For example, for $n=5, n'=4,d=2$, we have that $\mathbf{A}_{\mathrm{exc}} = \big\{\{1,5\},\{2,5\},\{3,5\},\{4,5\}\big\}$ consisting of 4 pairs.}\(d\)-tuples 
\begin{equation}
\label{eq:A_dis}
\mathbf{A}_{\mathrm{exc}} \triangleq \mathbf{A}_{n,d} \setminus \mathbf{A}_{n',d}.
\end{equation}
Thus, we aim to distribute the $\mathbf{A}_{\mathrm{exc}}$ that contain excluded files (from $\mathcal{E}$) into the groups \(N'={ f \choose d}\) formed by their non-excluded elements (see Section \ref{subsec: case1}). 
Any \(d\)-tuple \(\mathbf{t} \in \mathbf{A}_{\mathrm{exc}}\) will have an arbitrary number \(m_{\mathbf{t}}=|\mathbf{t}  \cap \mathcal{E}|\) of components/elements from the excluded file-index set $\mathcal{E}$, and it will have  \(d-m_{\mathbf{t}} = |\mathbf{t}  \cap [n']|\) elements from the rest. 
It is easy to see that  \( m_{\mathbf{t}}\in [1, \min\{d,g\}]\) and thus that \(d-m_{\mathbf{t}} \in [\max\{d-g,0\},d-1]\).  Whenever there is no ambiguity, we will henceforth revert to the simpler notation \(m\) instead of $ m_{\mathbf{t}}$.  
For every \(m\in [1, \min\{d,g\}]\), we define the set
 \begin{equation} 
    \label{eq:R_m,beta}
    \mathbf{R}_{m,\beta} \triangleq \left\{ \mathbf{t} \in \mathbf{A}_{\mathrm{exc}} \;\middle|\; |\mathcal{B}(\mathbf{t})| = \beta  ,\ |\mathbf{t}  \cap \mathcal{E}|=m \right\}
    \end{equation}
which describes the \(d\)-tuples \( \mathbf{t} \) that intersect exactly \( \beta \) families and contain \(m\) excluded elements from $\mathcal{E}$. Notice that \(\beta\) can take values in the range \(\big[\lceil\frac{d-m}{s_0}\rceil, d-m\big]\). If \(m=d\leq g\), then \(\beta =0\), which means that all the entries of \(\mathbf{t}\) are from \(\mathcal{E}\). Let us now partition \(\mathbf{A}_{\mathrm{exc}}\) as follows
\( \mathbf{A}_{\mathrm{exc}}=\bigcup_{m=1}^{\min\{d, g\}}\bigcup_{\beta=\lceil\frac{d-m}{s_0}\rceil}^{d-m} \mathbf{R}_{m,\beta}. \)
    For each $\mathcal{I}\in\binom{[f]}{\beta}$, let us now define 
 \begin{equation}
\mathbf{R}_{\beta,\mathcal{I}}
 \triangleq \big\{\mathbf{t}\in\mathbf{A}_{\mathrm{exc}}\ \mid \mathcal{B}(\mathbf{t})=\mathcal{I},\ 
 \mathbf{t}\in\bigcup_{m=1}^{\min\{d-\beta,g\}}\mathbf{R}_{m,\beta}\big\}
 \end{equation}
 to be the set of all d-tuples $ \mathbf{t} \in \mathbf{A}_{\mathrm{exc}}$ that intersect exactly all families in $\mathcal{I}$, where in the above, $\mathcal{B}(\mathbf{t})$ denotes the set of families that $\mathbf{t}$ intersects. Let us now also define 
\(  \mathbf{R}_{\beta}\triangleq\bigcup_{\mathcal{I}\in\binom{[f]}{\beta}}\mathbf{R}_{\beta,\mathcal{I}} \subset \mathbf{A}_{\mathrm{exc}} \) to be the set of all excluded \(d\)-tuples that meet exactly $\beta$ families.
Furthermore, directly by applying the established ranges of parameters \(m\) and \(k\), we can conclude that the range of \(\beta \in [\beta_{\mathrm{min}}, \beta_{\mathrm{max}}]\), is defined by \(
    \beta_{\min} \;\triangleq\; \lceil \frac{d-\min\{d,g\}}{s_0} \rceil
            \;=\;\lceil \frac{\max\{0,d-g\}}{s_0} \rceil
\)
and $\beta_{\max} \;\triangleq\; d-1.$

   Our next step involves going through the range of $\beta$. For each \(\beta \in [\beta_{\mathrm{min}}, \beta_{\mathrm{max}}]\), we partition each time the set \(\mathbf{R}_{\beta}\) into \( N'={f \choose d} \) groups.  This partitioning is described in detail in Appendix B.E of \cite{IC}. In particular, let us first recall that each group is labeled by a \(\sigma \in {[f] \choose d}\). For each such $\sigma$, there exist \({d \choose \beta}\) different subsets \(\mathcal{I} \subset\sigma\) with cardinality \(\beta\). For each \(\mathcal{I}\subset \sigma\), the set \(\mathbf{R}_{\beta, \mathcal{I}, \sigma}\) collects all \(d\)-tuples in \(\mathbf{R}_{\beta, \mathcal{I}}\) associated to group $\sigma$. We then form the union and define
   \begin{equation}
       \label{eq: phi_exc}
       \mathbf{\Phi}_{\sigma}^{\mathrm{(exc)}}\triangleq\bigcup\nolimits_{\beta=\beta_{\mathrm{min}}}^{\beta_{\mathrm{max}}}\bigcup\nolimits_{\mathcal{I}\subset \sigma}\mathbf{R}_{\beta, \mathcal{I}, \sigma}.
   \end{equation}
Combining \eqref{eq: step 1 portion} and 
\eqref{eq: phi_exc}, we get the subfunctions (\(d\)-tuples) allocated to worker \(\sigma\), where \(\sigma \in {[f]\choose d}\), as follows.
 \begin{equation}
     \label{eq:SubfunctionCase1}
     \mathbf{\Phi}_{\sigma}\triangleq\mathbf{\Phi}_{\sigma}^{(\mathrm{full})}\cup \mathbf{\Phi}_{\sigma}^{\mathrm{(com)}} \cup \mathbf{\Phi}_{\sigma}^{\mathrm{(exc)}}.
 \end{equation}
Finally, the partition of $\Adn$ into \(N'\) groups is described by 
\begin{equation}
\label{eq: full partition}
    \Adn= \bigcup\nolimits_{\sigma \in {[f] \choose d}}\mathbf{\Phi}_{\sigma}.
\end{equation}
We continue with the following lemma.

\begin{lemma}
\label{lem:lemma_pi_g>0}
For given \( n \), \( d\), and \(N\), the IC design in Case 2 (Section~\ref{subsec: Case 2}) achieves
\[
\pi \le s_0\cdot d +g.
\]
\end{lemma}
\begin{IEEEproof}
The proof is direct by noting that worker $\sigma$ receives all files in families $\sigma$ plus, at worst case, the entire set $\mathcal{E}$. 
\end{IEEEproof}

\subsection{Extension of the Partition from \(N'\) Groups to \(N\) Groups}
\label{subsec: N' to N}
Recall (cf.~\eqref{eq: full partition}) that we have already partitioned \(\mathbf{A}_{n,d}\) into \(N' = \binom{f}{d}\) disjoint groups \(
 \mathbf{\Phi}_{\sigma_1},\dots,\mathbf{\Phi}_{\sigma_{N'}},
\ \  \sigma_1,\dots,\sigma_{N'}\in\binom{[f]}{d}.   
\)
We will now redistribute the \(d\)-tuples of these $N'$ groups across all existing $N$ groups. Towards this, let us assume that the indices \(\sigma_1,\dots,\sigma_{N'} \) are in lexicographic order and, in order to ease notation, let us rename the corresponding \(N'\) groups by their lexicographic position, as follows \(  \mathbf{\Phi}_1,\dots,\mathbf{\Phi}_{N'}  \)
where in particular, \(\mathbf{\Phi}_b=\mathbf{\Phi}_{\sigma_b}\) for \(b \in[N']\). 
Recalling that there are \(N\ge N'\) actual groups, let us first define the variables
\begin{equation}
\label{eq: q,p,r}
  q\triangleq\big\lfloor\frac{N}{N'}\big\rfloor,\qquad
p\triangleq\big\lceil\frac{N}{N'}\big\rceil,\qquad
r\triangleq N\bmod N'
\end{equation}
thus noting that \(
N=qN'+r,\) where \(p=q\ \text{if }r=0\), and  \(p=q+1\ \text{if }r>0.\) 
At this point, we proceed with the first step of dividing the \(d\)-tuple set of each of the first $N'$ groups into different parts, and then with the second step of redistributing some of these parts to fill up the empty $N-N'$ groups.
\paragraph*{Step 1 -- Dividing the \(d\)-tuples of each of the first $N'$ groups}   
For each \(b\in[N']\), we define the number of parts 
\begin{equation}
s_b\triangleq\begin{cases}
p & \text{if } 1\le b\le r,\\[4pt]
q & \text{if } r<b\le N'
\end{cases}
\end{equation}
and we split each \(\mathbf{\Phi}_b\) into \(s_b\) disjoint sub-parts using lexicographic ordering that yields slicing of equal sizes, plus or minus $1$, where we naturally keep track of the exact size of each sub-part. We denote these sub-parts by \(\mathbf{\Phi}_b^{(0)},\mathbf{\Phi}_b^{(1)},\dots,\mathbf{\Phi}_b^{(s_b-1)},\) where \(\mathbf{\Phi}_b=\bigcup_{b'=0}^{s_b-1}\mathbf{\Phi}_b^{(b')}. \)
\paragraph*{Step 2 -- Extending to \(N\) groups} 
We then relabel these sub-parts to obtain the desired \(N\) groups.  
We define the new \(N\) groups \(\mathbf{\Phi}_1,\dots,\mathbf{\Phi}_N\) by the indexing rule
\begin{equation}
\label{eq: extenstion}
  \mathbf{\Phi}_{\,b + b'N'} \;\triangleq\;\mathbf{\Phi}_b^{(b')},
\ \text{for } b\in[N'],\; b'\in\{0,\dots,s_b-1\}.  
\end{equation}
We conclude this section with the following lemma.
\begin{lemma} 
\label{lem:N,N'}
For any \(n,d,N\), the IC design uses \(N' = \binom{f}{d}\), where 
\(f = \max \left\{ r \in \mathbb{Z}^{+} \,\middle|\, \binom{r}{d} \le N \right\}\) (cf.~\eqref{eq: largest r}), and guarantees that
\[
\frac{N}{N'} < d+1 \le 2^d .
\]
\end{lemma}
\begin{IEEEproof}
Directly from the definitions of \(f\) and \(N'\), we note that
\( N<\binom{f+1}{d}\) and \( N'=\binom{f}{d} \), and thus \(    \frac{N}{N'}<\frac{\binom{f+1}{d}}{\binom{f}{d}}=\frac{f+1}{f+1-d}. \)
Since the function \(\frac{x}{x - d}\) is decreasing on \([d+1, \infty)\), we conclude that 
\(\frac{N}{N'}\le\max_{f \ge d} \frac{f+1}{f+1-d} \le d+1\le 2^d.\)
\end{IEEEproof}

\section{A Lower Bound on \(\pi^\star\)}
\label{sec:lb}
A converse bound on the communication cost can be derived by observing that a worker with $\pi$ files can compute at most $\binom{\pi}{d}$ subfunctions. To cover all $\binom{n}{d}$ tasks with $N$ workers, we must have
\begin{equation}
\label{eq:lower0}
\binom{n}{d}\leq \sum_{b=1}^N \binom{|\alpha(\mathbf{\Phi}_b)|}{d} \leq N\binom{\pi}{d}
\end{equation}
since \(\pi =\max_{b\in [N]}|\alpha(\mathbf{\Phi}_b)| \).

Using the inequality $\frac{\pi-i}{n-i} \le \frac{\pi}{n} $ for \(1\le i\le d-1\) in \eqref{eq:lower0}, we obtain the following lower bound on \(\pi\)
\begin{equation}
    \label{eq:lower}
    \pi\geq \frac{n}{N^{1/d}}.
\end{equation}
For any possible \(\pi\), the \eqref{eq:lower} holds. Consequently, for the optimal \(\pi\), denoted by \(\pi^{\star}\), we have 
\begin{equation}
    \label{eq:lower_bound}
    \pi^{\star}\geq \frac{n}{N^{1/d}}.
\end{equation}
This lower bound represents the \emph{packing radius} of the hypergraph. While Steiner systems achieve this bound with equality (where $\pi$ is exactly the block size), their non-existence for most $N$ forces us to seek approximate designs that still respect this $N^{-1/d}$ scaling (please see Appendix B.A in \cite{IC} for more details).

\section{Proof of Theorem \ref{th: main}}
\label{Proof:Main}
From the achievable scheme discussed in Section~\ref{subsec: case1}, we have $\pi = s \cdot d = \frac{n}{f} \cdot d$. Using the simple bound $N'=\binom{f}{d} \le (\frac{e\cdot f}{d})^d$, we can conclude that $f \ge \frac{d}{e} N'^{1/d}$, which directly yields 
\begin{align}
\label{eq: prof,th1,1}
    \pi \le \frac{nd}{\frac{d}{e} N'^{1/d}}= \frac{N^{1/d}}{N'^{1/d}}\cdot\frac{nd}{\frac{d}{e} N^{1/d}} = \frac{2e\cdot n}{N^{1/d}}
\end{align}
where the last step follows from Lemma \ref{lem:N,N'}.  Then from Lemma~\ref{lem:lemma_pi_g>0}, we conclude that $\pi = s_0 \cdot d+g$. Similarly, we can show that
\begin{equation}
\label{eq: prof, th1,3}
  s_0\cdot d\le \frac{2e\cdot n}{N^{1/d}}.  
\end{equation}
This, combined with $g = n - s_0 \cdot f$, directly yields
\begin{align}
    g = n - \left( \left\lfloor \frac{n}{f+ d} \right\rfloor + 1 \right) \cdot f 
\leq n - \left\lceil \frac{n}{f + d} \right\rceil \cdot f \\
\leq n - \frac{n}{f + d} \cdot f
= \frac{n \cdot d}{f + d}
\end{align}
and since \( \frac{n}{f + d} < \left\lfloor \frac{n}{f + d} \right\rfloor + 1 = s_0 \), we can directly conclude that 
\(g \leq \frac{n \cdot d}{f + d} \leq s_0 \cdot d.\) Combining this with \eqref{eq: prof, th1,3}, we get $\pi\le \frac{4e \cdot n}{N^{1/d}} $. Finally, applying the converse in \eqref{eq:lower_bound} shows that for any \((n,d,N \le (\tfrac{9}{10}\sqrt{\frac{n}{d}})^d)\), we have \(\pi/\pi^\star \le 4e\), and hence the scaling law \(\pi \asymp n/N^{1/d}\) is optimal.

\section{Comparison with Steiner Systems}
The advantage of the IC design over Steiner systems lies in its flexibility with respect to column sizes. A Steiner system \(S(d,\pi,n)\) requires the number of blocks \(\binom{n}{d}/\binom{\pi}{d}\) to be exactly equal to \(N\), which severely restricts its applicability. In contrast, the IC design fixes \(N\) and \(n\), and then effectively determines the optimal clique size \(f\) (and hence \(\pi\)) compatible with the available workers. By allowing file assignments to overlap according to a family-based interlaced clique structure, rather than a rigid block structure, the IC design guarantees the existence of a valid partition for any \(N \le (\tfrac{9}{10}\sqrt{\frac{n}{d}})^d\).
As we have discussed, while a Steiner system would achieve \(\pi \asymp n/N^{1/d}\) with \(\delta=1\), it may not exist for a given \(N\). The IC design attains the same communication scaling with \(\delta \le 4\), offering a practical trade-off that represents a controlled increase in computation imbalance in exchange for universal applicability.

\section{Conclusion}
\label{sec:conclusion}

This paper addressed a fundamental file and task allocation problem in distributed computing. We highlighted the theoretical optimality of Steiner systems while exposing their practical limitations due to sparsity. The proposed Interweaved-Cliques (IC) design was shown to bridge this gap, offering a deterministic and universally applicable allocation scheme. By achieving order-optimal communication cost $\pi \asymp n/N^{1/d}$ and bounded computation balance, the IC design provides a robust solution for deploying large-scale distributed function evaluations without the rigid constraints of classical combinatorial designs.

\IEEEtriggeratref{21}
\bibliographystyle{IEEEtran}
\bibliography{references_}




\end{document}